\documentclass{article}
\usepackage{amsmath}
\usepackage{amssymb}
\usepackage{graphicx}

\newtheorem{theorem}{Theorem}

\newtheorem{definition}[theorem]{Definition}

\newtheorem{lemma}[theorem]{Lemma}

\newtheorem{proposition}[theorem]{Proposition}

\newenvironment{proof}[1][Proof]{\textbf{#1.} }{\ \rule{0.5em}{0.5em}}

\begin{document}

\author{Fabio Bellini \thanks{%
Dipartimento di Metodi Quantitativi, Universit\'{a} di Milano Bicocca, Italy.
E-mail: fabio.bellini@unimib.it}, Franco Pellerey \thanks{%
Dipartimento di Matematica, Politecnico di Torino, Italy. E-mail:
franco.pellerey@polito.it}, Carlo Sgarra\thanks{%
Dipartimento di Matematica, Politecnico di Milano, Italy. E-mail:
carlo.sgarra@polimi.it} and Salimeh Yasaei Sekeh \thanks{%
Department of Statistics, Ferdowsi University of Mashhad, Iran. E-mail:
sayasaei@yahoo.com}}
\title{Comparison Results for Garch Processes}
\maketitle

\begin{abstract}
We consider the problem of stochastic comparison of general Garch-like processes, for different parameters and different distributions of the innovations.
We identify several stochastic orders that are propagated from the innovations to the Garch process itself, and discuss their interpretations.
We focus on the convex order and show that in the case of symmetric innovations it is also propagated to the cumulated sums of the Garch process.
More generally, we discuss multivariate comparison results related to the multivariate convex and supermodular order. Finally we discuss ordering with respect to the parameters in the Garch (1,1) case.

\end{abstract}

\section{Introduction}

An extensive literature is available on applications of stochastic orders to
finance and insurance markets. The implications of stochastic orders for
derivative pricing and risk management are relevant.\ The increasing
dependence of european option prices by the riskiness of the underlying it
is a well known property for basic models like that of Black-Merton-Scholes,
in which riskiness is expressed in terms of the logreturns distribution
variance: the uncertainty is quantified there through the dispersion around
the expected value and the distribution functions can be ordered according
to their "peakedness"; the larger is the dispersions, the higher the option
prices. This very elementary and intuitive observation for simple models
become more involved when turning attention to more complex models, where a
more rigorous approach is necessary in order to avoid wrong conclusions.

The Black-Merton-Scholes model is nowadays considered fairly inadequate to
describe the asset price dynamics; several empirical facts cannot be
explained on the basis of this model: some statistical features exhibited by
logreturns like fat tails, volatility clustering, aggregational Gaussianity
and the so-called leverage effect are completely outside of the prevision
properties of the Black-Merton-Scholes model. Moreover a very relevant
phenomenon exhibited by option prices, the "volatility smile" (and its term
structure) cannot be explained on this model basis. In order to provide a
more satisfactory description several different models have been introduced.
Some of these models introduce a stocastic dependence in volatility and/or
jumps in asset logreturns (and/or in volatility) dynamics both in continuous
and discrete time setting.

Among discrete time models introduced in order to remove some of the
Black-Merton-Scholes model drawbacks, the class of Autoregressive
Conditioned Heteroschedastic (ARCH) models introduced by Engle \cite{Engle}
and their general extension (Garch models) proposed by Bollerslev in \cite%
{Bollerslev} have risen considerable interest.

Several results related to stochastic orders are available for the
continuous time models class: in \cite{BergRusch}, where a systematic
investigation on semimartingale models is performed; the models considered
there include the Heston and the Barndorff-Nielsen and Shephard models. In
\cite{Moeller}\ T. M{\o }ller provides some results on stochastic orders in
a dynamic reinsurance market where the traded risk process is driven by a
compound Poisson process and the claim amount is unbounded. Stochastic order
\ properties have been used to obtain bounds for option prices in incomplete
markets; the literature focused on this subject is quite extensive and we
just mention the papers by El Karoui et al. \cite{El Karoui}, by Bellamy and
Jeanblanc, and by Gushchin and Mordecki \cite{Gushchin}.

The purpose of the present paper is to present a systematic investigation of
stochastic orders propagation in a Garch\ context.

Comparison with stochastic orders in incomplete market models can give rise
to different classes of problems: first can be considered the comparison of
models under the same probability measure but with different parametric
specification, second it can be examined the problem of comparing the same
model under different probability measures; as a matter of fact, when
markets are incomplete, there are several probability measures equivalent to
the historical one, under which the dynamics of prices can be given. In this
paper we shall focus on the first class of problems mentioned: we shall
provide a systematic comparison of logreturns and then of prices when the
model parameters change, but the dynamics is specified under the same
probability measure. In a Garch context the parameters entering into play
are three parameters assuming a numerical value and the innovations, which
are random IID variables for which the density function is assigned. We just
mentioned that stochastic order results have important implications on
option pricing and this holds true also in a Garch context: in particular,
convex order relations on logreturn sums imply increasing convex order
relations on the underlying price, hence on european call option prices and
this can be considered the main relevance of our results from an application
viewpoint.

We like to present a numerical illustration in order to motivate our
investigation. Fig 1 compare the densities of the logreturn sums
in a Garch (1,1) model with respect to variations in the parameters $\alpha
_{0},\alpha _{1},\beta _{1}$; in the continuous line all parameters assume
the value 0.2, while the other lines represent the same density, but with
parameters $\alpha_{0}=0.5$,$\alpha _{1}=0.5$ and $\beta _{1}=0.5$ respectively
(the parameter assuming the value 0.5 in indicated by a capital letter in each
curve caption). The sums include the first 50 terms of the logreturn sequence.

\bigskip
\begin{figure}
\centering \includegraphics[width=1.0\textwidth]{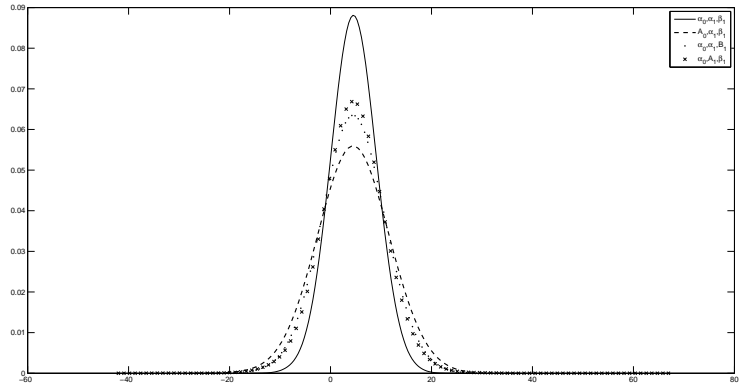}\\
  \caption{Comparison of Logreturn Sum Densities}\label{fig1}
\end{figure}

The innovations, assumed to be independent and identically distributed,
and the initial variance, are chosen as standard Gaussian random variables.
It is evident how the usual real numbers order relation between parameters implies
some ordering on logreturn sums; it then arises as a natural question to ask if
this simple remark can be made rigorous and if this conclusion can be cast
into a more general framework considering different kind of stochastic
orders and comparing stochastic order propagation from innovations to
logreturns and to logreturns sums.

In the following section we present briefly the Garch models and an
auxiliary lemma. In Section 3 we present the univariate stochastic comparisons
for logreturns in a Garch setting. In Section 4 we provide
results on some related orders, while in Section 5 the convex order
propagation of the logreturn sums is investigated together with its
implications on asset and european call option prices, while their
multivariate convex order propagation is the subject of Section 6 . We focus
our attention on the Garch(1,1) case in the last section.

The main results provided in this paper require the assumption of a symmetric
probability density for innovations.
The extension of the stochastic order propagation analysis presented here
to the case of non-symmetric innovation densities would be of great interest:
both the convex order propagation result and the comparison with respect to
parameter variations will be the subject of future investigation in this more general setting.
Moreover the identification of some convex multivariate order which naturally propagate
from innovations to logreturns is another target on which our research interest
will be focused. These will be the subject of our future work.

\section{General Garch models}

We consider Garch models of two different very general forms; the first
model (M1) is:
\begin{equation}
\left\{
\begin{array}{c}
X_{n}=\sigma _{n}\varepsilon _{n}\text{, }n=0,1,.. \\
\varepsilon _{n}\text{ }\bot \text{ }\sigma _{n}\text{, }E[\varepsilon
_{n}]=0 \\
\sigma _{n+1}=f^{I}(\left\vert \varepsilon _{n}\right\vert ,\sigma _{n})%
\end{array}%
\right.  \label{garch1}
\end{equation}%
with $f^{I}:\mathbb{R}_{+}^{2}\rightarrow \mathbb{R}_{+}$
increasing and componentwise convex (ccx for brevity).
\newline
The second model\ (M2) is%
\begin{equation}
\left\{
\begin{array}{c}
X_{n}=\sigma _{n}\varepsilon _{n}\text{, }n=0,1,.. \\
\varepsilon _{n}\text{ }\bot \text{ }\sigma _{n}\text{, }E[\varepsilon
_{n}]=0 \\
\sigma _{n+1}^{2}=f^{II}(\varepsilon _{n}^{2},\sigma _{n}^{2})%
\end{array}%
\right.  \label{garch2}
\end{equation}%
with $f^{II}:\mathbb{R}_{+}^{2}\rightarrow \mathbb{R}_{+}$ increasing and ccx.
In both cases the innovations $\varepsilon _{n}$
are independent and identically distributed (i.i.d.). When discussing the
propagation of variability orderings, the normalization requirement $%
E[\varepsilon _{n}^{2}]=1$ will be dropped. The difference between model M1
and model M2 is that in the first case the recursive dynamics is defined in
terms of the volatility $\sigma _{n}$, while in the second it is defined in
terms of the variance $\sigma _{n}^{2}$. \newline
The usual Garch (1,1) model is a particular case of both M1 and M2, and is
defined as follows:%
\begin{equation}
\left\{
\begin{array}{c}
X_{n}=\sigma _{n}\varepsilon _{n}\text{, }n=0,1,.. \\
\sigma _{n+1}^{2}=\alpha _{0}+\alpha _{1}X_{n}^{2}+\beta _{1}\sigma _{n}^{2}%
\end{array}%
\right.  \label{garch3}
\end{equation}%
with $\alpha _{0},\alpha _{1},\beta _{1}>0$ and $\alpha _{1}+\beta _{1}<1$,
in order to guarantee covariance stationarity. Both models start with a
possibly random $\sigma _{0}>0$, by drawing a random $\varepsilon _{0}$.%
\newline
The recursive equations for $\sigma _{n+1}$ and $\sigma _{n+1}^{2}$ are
examples of "stochastic recurrences" in the sense of Chapter 4 of \cite{ms}.
For the explicit expression of their solutions, we introduce the following
notations:%
\begin{align}
\sigma _{n+1} &:=g_{n+1}^{I}(\sigma _{0},\left\vert \varepsilon
_{0}\right\vert ,...,\left\vert \varepsilon _{n}\right\vert )
\label{gienne2} \\
\sigma _{n+1}^{2} &:=g_{n+1}^{II}(\sigma _{0}^{2},\varepsilon
_{0}^{2},...,\varepsilon _{n}^{2})  \notag
\end{align}%
As in \cite{ms}, we have the following:

\begin{lemma}
\label{gienneL}Let $g_{n+1}^{I}$, $g_{n+1}^{II}:%
\mathbb{R}
_{+}^{n+2}\rightarrow
\mathbb{R}
_{+}$ be defined as in (\ref{gienne2}). Then $g_{n+1}^{I}$ and $g_{n+1}^{II}$ are
increasing and componentwise convex.
\end{lemma}

\begin{proof}
We have clearly $g_{n+1}^{I}(\sigma _{0},\left\vert \varepsilon
_{0}\right\vert ,...,\left\vert \varepsilon _{n}\right\vert
)=f^{I}(\left\vert \varepsilon _{n}\right\vert ,...,f^{I}(\left\vert
\varepsilon _{0}\right\vert ,\sigma _{0}))$ and $g_{n+1}^{II}(\sigma
_{0}^{2},\varepsilon _{0}^{2},...,\varepsilon _{n}^{2})=f^{II}(\varepsilon
_{n}^{2},...,f^{II}(\varepsilon _{0}^{2},\sigma _{0}^{2}))$; since $f^{I}$
and $f^{II}$ are increasing and ccx, then also their compositions in (\ref%
{gienne2}) are increasing and ccx (see Meester and Shanthikumar \cite{Meester}
for this last assertion and further properties of increasing and ccx functions).
\end{proof}

\section{Univariate comparisons of $X_{n}$}

The aim of this section is to establish comparison results for $X_{n}$ when
the distributions of the innovations are changed from $\varepsilon _{k}$ to $%
\widetilde{\varepsilon }_{k}$. In order to establish these results, the
assumption that the innovations are identically distributed is not necessary
(while the independency assumption is essential). In the following theorems
only the distribution of a single innovation $\varepsilon _{k}$ will be
changed, and the impact of this change on $X_{n}$ will be investigated.

We recall the definitions of the basic stochastic orderings:
\begin{definition}
The random variable $X$ is said to be smaller than $Y$ in the usual stochastic order [convex order, increasing convex order],
denoted by $X\leq_{st}Y$ [$X\leq_{cx}Y, \ X\leq_{icx}Y$], if $E[\phi(X)] \leq E[\phi(Y)]$ for all increasing
[convex, increasing convex] functions $\phi:{\mathbb R}\to {\mathbb R}$ for which the expectations exist.
\end{definition}

We will see that in the general context of models M1 and M2 the orderings
that are naturally propagated from the innovations $\varepsilon _{k}$ to $%
X_{n}$ are the $\leq _{st}$ and the $\leq _{icx}$ ordering \textit{between
absolute values or squared variables}. This clearly completely modifies
their interpretation; in particular, in the next section we will see that
the $\leq _{st}$ ordering between absolute values or squares can be
interpreted as a \textit{variability ordering}, while the $\leq _{icx}$
ordering between absolute values or squares can be interpreted as a \textit{%
kurtosis} ordering.

In order to establish these results, we proceed in two steps: first we
consider the volatilities $\sigma _{n}$ and then the variables $X_{n}$. The
first step is an immediate consequence of Lemma \ref{gienneL}:

\begin{theorem}
\label{compsigma}\textbf{Comparisons of }$\sigma _{n}$\textbf{\ and }$\sigma
_{n}^{2}$\newline
a) Let $\sigma _{n+1}$ be as in (\ref{garch1}) and $\left\vert \varepsilon
_{k}\right\vert \preceq _{st}\left\vert \widetilde{\varepsilon }%
_{k}\right\vert $; it follows that $\sigma _{n+1}\preceq _{st}\widetilde{%
\sigma }_{n+1}$.\newline
b) Let $\sigma _{n+1}$ be as in (\ref{garch1}) and $\left\vert \varepsilon
_{k}\right\vert \preceq _{icx}\left\vert \widetilde{\varepsilon }%
_{k}\right\vert $; it follows that $\sigma _{n+1}\preceq _{icx}\widetilde{%
\sigma }_{n+1}$.\newline
c) Let $\sigma _{n+1}^{2}$ be as in (\ref{garch2}) and $\varepsilon
_{k}^{2}\preceq _{st}\widetilde{\varepsilon }_{k}^{2}$; it follows that $%
\sigma _{n+1}^{2}\preceq _{st}\widetilde{\sigma }_{n+1}^{2}$.\newline
d) Let $\sigma _{n+1}^{2}$ be as in (\ref{garch2}) and $\varepsilon
_{k}^{2}\preceq _{icx}\widetilde{\varepsilon }_{k}^{2}$; it follows that $%
\sigma _{n+1}^{2}\preceq _{icx}\widetilde{\sigma }_{n+1}^{2}$.
\end{theorem}

\begin{proof}
Since from Lemma \ref{gienneL} in model M1 $\sigma _{n+1}=g_{n+1}^{I}(\sigma
_{0},\left\vert \varepsilon _{0}\right\vert ,...,\left\vert \varepsilon
_{n}\right\vert )$ with $g_{n+1}^{I}$ increasing and ccx, item a) and b)
follows respectively fromTheorem 1.A.3 in \cite{ss} and Theorem 4.A.15 in
\cite{ss}. Similarly, since from Lemma \ref{gienneL} in model M2 $\sigma
_{n+1}^{2}=g_{n+1}^{II}(\sigma _{0}^{2},\varepsilon _{0}^{2},...,\varepsilon
_{n}^{2})$ with $g_{n+1}^{II}$ increasing and ccx, from the same theorems we
get c) and d). \medskip
\end{proof}

The comparison results for $\sigma _{n}$ and $\sigma _{n}^{2}$ lead to the
following comparisons of the variables $X_{n}$:

\begin{theorem}
\label{compX}\textbf{Comparisons of }$X_{n}$\newline
a) Let $X_{n}$ be as in (\ref{garch1}) and $\left\vert \varepsilon
_{k}\right\vert \preceq _{st}\left\vert \widetilde{\varepsilon }%
_{k}\right\vert $; it follows that $\left\vert X_{n}\right\vert \preceq
_{st}\left\vert \widetilde{X}_{n}\right\vert $.\newline
b) Let $X_{n}$ be as in (\ref{garch1}) and $\left\vert \varepsilon
_{k}\right\vert \preceq _{icx}\left\vert \widetilde{\varepsilon }%
_{k}\right\vert $; it follows that $\left\vert X_{n}\right\vert \preceq
_{icx}\left\vert \widetilde{X}_{n}\right\vert $.\newline
c) Let $X_{n}$ be as in (\ref{garch2}) and $\varepsilon _{k}^{2}\preceq _{st}%
\widetilde{\varepsilon }_{k}^{2}$; it follows that $X_{n}^{2}\preceq _{st}%
\widetilde{X}_{n}^{2}$.\newline
d) Let $X_{n}$ be as in (\ref{garch2}) and $\varepsilon _{k}^{2}\preceq _{icx}%
\widetilde{\varepsilon }_{k}^{2}$; it follows that $X_{n}^{2}\preceq _{icx}%
\widetilde{X}_{n}^{2}$.
\end{theorem}

\begin{proof}
Since $\left\vert X_{n}\right\vert =\sigma _{n}\left\vert \varepsilon
_{n}\right\vert $ and $X_{n}^{2}=\sigma _{n}^{2}\varepsilon _{n}^{2}$, with $%
\sigma _{n}$ independent from $\varepsilon _{n}$, items a), and c) follow
from Theorem 1.A.3 (b) in \cite{ss}. Similarly item b) and d) follow from
Theorem 4.A.15 in \cite{ss}.\medskip
\end{proof}

A natural question that arises at this point is if also the convex order is
propagated, that is if $\varepsilon _{k}\preceq _{cx}\widetilde{\varepsilon }%
_{k}\Rightarrow X_{n}\preceq _{cx}\widetilde{X}_{n}$. We prove that this is
indeed the case for model M1. We start with a simple lemma:

\begin{lemma}
\label{stconv}Let $\sigma $ and $\tilde{\sigma}$ be nonnegative, with $%
\sigma \preceq _{st}\tilde{\sigma}$. Let $\varepsilon $ be independent from $%
\sigma $ and $\tilde{\sigma}$, with $E\left[ \varepsilon \right] =0$; then $%
\sigma \varepsilon \preceq _{cx}\tilde{\sigma}\varepsilon $.
\end{lemma}

\begin{proof}
We can construct identically distributed copies of $\sigma $ and $\tilde{%
\sigma}$ on the same probability space, such that $\sigma \leq \tilde{\sigma}
$ a.s. Then for each realization of $\sigma $, we have that $\sigma
\varepsilon \preceq _{cx}\tilde{\sigma}\varepsilon $. From Theorem 3.A.12 in
\cite{ss} it follows that $\sigma \varepsilon \preceq _{cx}\tilde{\sigma}%
\varepsilon $.\medskip
\end{proof}

\begin{theorem}
\label{propconv}\textbf{Propagation of convex order}\newline
Let $X_{n}$ be as in (\ref{garch1}) and $\varepsilon _{k}\preceq _{cx}%
\widetilde{\varepsilon }_{k}$; it follows that $X_{n}\preceq _{cx}\widetilde{%
X}_{n}$
\end{theorem}

\begin{proof}
First of all we remark that since $\varepsilon _{k}\preceq _{cx}\widetilde{%
\varepsilon }_{k}$, it follows that $\left\vert \varepsilon _{k}\right\vert
\preceq _{icx}\left\vert \widetilde{\varepsilon }_{k}\right\vert $. Indeed,
for each $\phi $ increasing and convex, the composition $\phi (\left\vert
...\right\vert )$ is convex; this implies that $E[\phi (\left\vert
\varepsilon _{k}\right\vert )]\leq E[\phi (\left\vert \widetilde{\varepsilon
}_{k}\right\vert )]$, that is $\left\vert \varepsilon _{k}\right\vert
\preceq _{icx}\left\vert \widetilde{\varepsilon }_{k}\right\vert $. From
Proposition \ref{compsigma} item b), it then follows that $\sigma
_{n+1}\preceq _{icx}\widetilde{\sigma }_{n+1}$. From Theorem 4.A.6 in \cite%
{ss} there exists a random variable $\overline{\sigma }_{n+1}$ with $\sigma
_{n+1}\preceq _{st}\overline{\sigma }_{n+1}\preceq _{cx}\widetilde{\sigma }%
_{n+1}$. By Lemma \ref{stconv}, we have that $\sigma _{n+1}\preceq _{st}%
\overline{\sigma }_{n+1}$ implies that $\sigma _{n+1}\varepsilon
_{n+1}\preceq _{cx}\overline{\sigma }_{n+1}\varepsilon _{n+1}$; on the other
hand $\overline{\sigma }_{n+1}\preceq _{cx}\widetilde{\sigma }_{n+1}$
implies that $\overline{\sigma }_{n+1}\varepsilon _{n+1}\preceq _{cx}%
\widetilde{\sigma }_{n+1}\varepsilon _{n+1}$. By transitivity we get $\sigma
_{n+1}\varepsilon _{n+1}\preceq _{cx}\widetilde{\sigma }_{n+1}\varepsilon
_{n+1}$. \medskip
\end{proof}

\section{The relevant orderings}

In the preceding section the orderings defined by $\left\vert X\right\vert
\preceq _{st}\left\vert Y\right\vert $, $X^{2}\preceq _{st}Y^{2}$, $%
\left\vert X\right\vert \preceq _{icx}\left\vert Y\right\vert $, $%
X^{2}\preceq _{icx}Y^{2}$ have arisen naturally.\newline
In order to better understand their meaning, in the following lemmas we
identify some necessary and sufficient conditions\textbf{\ }in the
continuous and symmetric case. We have the following:

\begin{lemma}
Let $X$ and $Y$ be symmetric with continuous distributions $F$ and $G$. The
following conditions are equivalent:\newline
a) $X^{2}\preceq _{st}Y^{2}$;\newline
b) $\left\vert X\right\vert \preceq _{st}\left\vert Y\right\vert $;\newline
c) $X\preceq _{peak}Y$, where $\preceq _{peak}$ is the peakedness ordering
introduced in \cite{birnbaum};\newline
d) $S^{-}(G-F)=1$ with sign sequence $+,-$, where $S^{-}(G-F)$ is the number
of intersections between $G$ and $F$ as defined in (1.A.18) of \cite{ss}.%
\newline
\end{lemma}

\begin{proof}
The equivalence of a) and b) is an immediate consequence of Theorem 1.A.3 in
\cite{ss}. The equivalence of b) and c) is the definition of the peakedness
ordering, while the equivalence between c) and d) follows from Theorem 3.D.1
in \cite{ss}.\medskip
\end{proof}

\begin{lemma}
Let $X$ and $Y$ be symmetric with continuous distributions $F$ and $G$. The
following conditions are equivalent:\newline
a) $X^{2}\preceq _{icx}Y^{2}$ \newline
b) $\int_{x}^{+\infty }\overline{F}(u)\cdot u\cdot du\leq \int_{x}^{+\infty }%
\overline{G}(u)\cdot u\cdot du$ for each $x\geq 0$, where $\overline{F}%
(u)=1-F(u)$ and $\overline{G}(u)=1-G(u)$\newline
c) $E[(X^{2}-k)^{+}]\leq E[(Y^{2}-k)^{+}]$ for each $k\geq 0$
\end{lemma}

\begin{proof}
Under our hypothesis $F_{X^{2}}(t)=2F(\sqrt{t})-1$ and $\overline{F}%
_{X^{2}}(t)=2-2F(\sqrt{t})$, for $t\geq 0$. The equivalence of a) and b)
follows then from Theorem 4.A.2 in \cite{ss} with a simple change of
variable. The equivalence between a) and c) is also a consequence of Theorem
1.5.7 in \cite{ms}\medskip
\end{proof}

The first lemma shows that for symmetric variables the orderings $\left\vert
X\right\vert \preceq _{st}\left\vert Y\right\vert $ and $X^{2}\preceq
_{st}Y^{2}$ are variability comparisons equivalent to the peakedness ordering,
that in this case boils down to item d), that is the validity of a single
cut condition between the distribution functions. In the typical econometric
applications these orderings are however not very relevant since the
innovations satisfy $E[\varepsilon _{k}^{2}]=1$, and hence $E[\varepsilon
_{k}^{2}]\preceq _{st}E[\widetilde{\varepsilon }_{k}^{2}]$ would imply $%
\varepsilon _{k}^{2}=_{st}\widetilde{\varepsilon }_{k}^{2}$.

In the normalized case the ordering $X^{2}\preceq _{icx}Y^{2}$ becomes
equivalent to $X^{2}\preceq _{cx}Y^{2}$; we prove a sufficient and a
necessary condition for it.

\begin{lemma}
Let $X$ and $Y$ be symmetric with continuous distributions $F$ and $G$ and
with $E[X^{2}]=E[Y^{2}]=1$. \newline
a) If the densities of $X$ and $Y$ cross $4$ times, with the density of $X$
being lower in the tails and in the center, and higher in the intermediate
region, then $X^{2}\preceq _{icx}Y^{2}$.\newline
b) If $X^{2}\preceq _{icx}Y^{2}$ and $X$ and $Y$ have finite fourth moments,
then $\beta _{2}(X)<\beta _{2}(Y)$, where $\beta _{2}$ is Pearson's kurtosis
coefficient.
\end{lemma}

\begin{proof}
a) Under our hypothesis $f_{X^{2}}(t)=\frac{f(\sqrt{t})}{\sqrt{t}}$ for $t>0$%
. Since $X$ and $Y$ are symmetrical, we have that the four intersection
points between the densities $f$ and $g$ are symmetrical with respect to
the origin. Hence the densities of $X^{2}$ and $Y^{2}$ cross in two points
and since $E[X^{2}]=E[Y^{2}]$ from Theorem 3.A.44 in \cite{ss} we have that $%
X^{2}\preceq _{cx}Y^{2}$.\newline
b) In our case $\beta _{2}(X)=E[X^{4}]$ and hence the thesis follows from
the definition of the convex order.\medskip
\end{proof}

This lemma shows that the comparison $X^{2}\preceq _{icx}Y^{2}$ can be
interpreted as a classical kurtosis ordering; the cut condition is usually
referred in the kurtosis ordering literature as a Dyson-Finucan condition
(see for example \cite{finucan} for a review).

\section{Convex comparisons for total logreturns}

In financial applications the variables $X_{n}$ typically represent
logreturns, that are additive quantities. The over-the-period total return
is given by $S_{n}=\sum\limits_{k=1}^{n}X_{n}$. It is therefore natural to
ask if some of the comparison results of Section 2 do extend to the
variables $S_{n}$. In this section we consider the case of the convex order,
that is, we wonder if $\varepsilon _{k}\preceq _{cx}\widetilde{%
\varepsilon }_{k}\Rightarrow S_{n}\preceq _{cx}\widetilde{S}_{n}$. The
problem is not trivial since $S_{n}$ cannot be expressed as a sum of
independent variables, so the standard results about convex ordering of sums
cannot be applied; we are able to prove a positive result in the case of
model M1 and for symmetric innovations. We start with a basic lemma:

\begin{lemma}
\label{BL}Let $\phi \in C^{2}(%
\mathbb{R}
)$ be convex and $g_{i}\in C^{2}(%
\mathbb{R}
)$ be convex and nonnegative. Let $a,b\in R$ and $P_{m}:=\left\{
-1,1\right\} ^{m}$. It follows that
\begin{equation*}
h(u)=\sum\limits_{\underline{p}\in P_{m}}\phi \left(
a+bu+\sum\limits_{i=1}^{m}p_{i}g_{i}(u)\right)
\end{equation*}
is convex.
\end{lemma}

\begin{proof}
We can compute%
\begin{equation*}
h^{\prime }(u)=\sum\limits_{\underline{p}\in P_{m}}\phi ^{\prime }\left(
a+bu+\sum\limits_{i=1}^{m}p_{i}g_{i}(u)\right) \cdot \left(
b+\sum\limits_{i=1}^{m}p_{i}g_{i}^{\prime }(u)\right)
\end{equation*}%
\begin{eqnarray*}
h^{\prime \prime }(u) &=&\sum\limits_{\underline{p}\in P_{m}}\phi ^{\prime
\prime }\left( a+bu+\sum\limits_{i=1}^{m}p_{i}g_{i}(u)\right) \cdot \left(
b+\sum\limits_{i=1}^{m}p_{i}g_{i}^{\prime }(u)\right) ^{2}+ \\
&&\sum\limits_{\underline{p}\in P_{m}}\phi ^{\prime }\left(
a+bu+\sum\limits_{i=1}^{m}p_{i}g_{i}(u)\right) \cdot \left(
\sum\limits_{i=1}^{m}p_{i}g_{i}^{\prime \prime }(u)\right)
\end{eqnarray*}%
The first term is positive; the second is given by%
\begin{equation*}
A_{m}=\sum\limits_{\underline{p}\in P_{m}}\phi ^{\prime }\left(
a+bu+\sum\limits_{i=1}^{m}p_{i}g_{i}(u)\right) \cdot \left(
\sum\limits_{i=1}^{m}p_{i}g_{i}^{\prime \prime }(u)\right)
\end{equation*}%
Let us denote with $\underline{P}$ a random vector with a discrete uniform
distribution on $P_{m}$; clearly $E[\underline{P}]=\underline{0}$, the
components of $\underline{P}$ are independent and
\begin{equation*}
A_{m}=2^{m}E\left[ \phi ^{\prime }\left( a+bu+\underline{g}(u)\cdot
\underline{P}\right) \left( \underline{g}^{\prime \prime }(u)\cdot
\underline{P}\right) \right]
\end{equation*}%
Since the functions $\phi ^{\prime }\left( a+bu+\underline{g}(u)\cdot
\underline{P}\right) $ and $\underline{g}^{\prime \prime }(u)\cdot
\underline{P}$ are componentwise increasing, from the covariance inequality
it follows that%
\begin{eqnarray*}
A_{m} &=&2^{m}E\left[ \phi ^{\prime }\left( \underline{g}(u)\cdot \underline{%
P}\right) \left( \underline{g}^{\prime \prime }(u)\cdot \underline{P}\right) %
\right] \geq \\
&\geq &2^{m}E\left[ \phi ^{\prime }\left( \underline{g}(u)\cdot \underline{P}%
\right) \right] E\left[ \left( \underline{g}^{\prime \prime }(u)\cdot
\underline{P}\right) \right] =0
\end{eqnarray*}
\end{proof}

We remark that in this lemma the smoothness requirements on $\phi $ and on
the $g_{i}$ can be dropped; we preferred this formulation in order to
simplify the proof. Since in this section we consider only model M1, we
define
\begin{equation*}
g_{n}(\sigma _{0},\varepsilon _{0},...,\varepsilon _{n-1}):=g_{n}^{I}(\sigma
_{0},\left\vert \varepsilon _{0}\right\vert ,...,\left\vert \varepsilon
_{n-1}\right\vert )\text{;}
\end{equation*}%
from Lemma \ref{gienneL}, it is clear that $g_{n}$ is even and ccx. We have

\begin{eqnarray}
S_{n} &=&X_{0}+X_{1}+...+X_{n}=\sigma _{0}\varepsilon _{0}+\sigma
_{1}\varepsilon _{1}+...+\sigma _{n}\varepsilon _{n}=  \notag \\
&&\sigma _{0}\varepsilon _{0}+g_{1}(\sigma _{0},\varepsilon _{0})\varepsilon
_{1}+...+g_{n}(\sigma _{0},\varepsilon _{0},...,\varepsilon
_{n-1})\varepsilon _{n}=  \label{sums} \\
&&S_{n}(\sigma _{0},\varepsilon _{0},...,\varepsilon _{n})  \notag
\end{eqnarray}%
The main problem in proving the propagation of convexity to the sums is that
$S_{n}$ is not a ccx functions of the innovations $\varepsilon _{k}$;
indeed, each $g_{k}$ in (\ref{sums}) is multiplied by a possibly negative
innovation $\varepsilon _{k}$. This prevents the applications of standard
results and requires the development of a specific technique based on Lemma %
\ref{BL}. The basic idea is that in the case of symmetric innovations it is
possible to restore the convexity by averaging over all the possible sign
changes, as in Lemma \ref{BL}. This will be done in a recursive way; we
start with the following:

\begin{lemma}
\label{convexity}Let $X_{n}$ and $S_{n}$ be as in (\ref{garch1}) and (\ref%
{sums}). Let $\phi $ be convex and $\varepsilon _{i}$ be symmetric. Then the
function
\begin{equation}
h(\sigma _{0},\varepsilon _{0},...,\varepsilon _{k}):=E_{\varepsilon
_{k+1},...,\varepsilon _{n}}[\phi (S_{n}(\sigma _{0},\varepsilon
_{0},...,\varepsilon _{n}))]  \label{acca}
\end{equation}%
\newline
is convex in $\varepsilon _{k}$ for each fixed value of $\sigma
_{0},\varepsilon _{0},...,\varepsilon _{k-1}$.
\end{lemma}

\begin{proof}
To avoid notational burdening we drop the arguments of the functions $g_{i}$%
. Since the innovations are symmetric and $g_{i}$ is even, we can write%
\begin{equation*}
E_{\varepsilon _{k+1},...,\varepsilon _{n}}[\phi (S_{n}(\sigma
_{0},\varepsilon _{0},...,\varepsilon _{n}))]=E_{\varepsilon
_{k+1},...,\varepsilon _{n}}[\phi (\sigma _{0}\varepsilon
_{0}+...+g_{n}\varepsilon _{n})]=
\end{equation*}%
\begin{equation*}
E_{\varepsilon _{k+1},...,\varepsilon _{n}}[\sum\limits_{\underline{p}\in
P_{n-k}}\phi (\sigma _{0}\varepsilon _{0}+...+p_{1}g_{k+1}\varepsilon
_{k+1}+...+p_{n-k}g_{n}\varepsilon _{n})1_{\left\{ \varepsilon _{k+1}\geq
0,...,\varepsilon _{n}\geq 0\right\} }]\text{.}
\end{equation*}%
Denoting by
\begin{equation*}
\overline{h}(\sigma _{0},\varepsilon _{0},...,\varepsilon
_{k},...,\varepsilon _{n})=\sum\limits_{\underline{p}\in P_{n-k}}\phi
(\sigma _{0}\varepsilon _{0}+g_{1}\varepsilon _{1}+...+g_{k}\varepsilon
_{k}+p_{1}g_{k+1}\varepsilon _{k+1}+...+p_{n-k}g_{n}\varepsilon _{n})\text{,}
\end{equation*}%
we have that
\begin{equation*}
h(\sigma _{0},\varepsilon _{0},...,\varepsilon _{k-1},\varepsilon
_{k})=E_{\varepsilon _{k+1},...,\varepsilon _{n}}[1_{\left\{ \varepsilon
_{k+1}\geq 0,...,\varepsilon _{n}\geq 0\right\} }\overline{h}(\sigma
_{0},\varepsilon _{0},...,\varepsilon _{k},...,\varepsilon _{n})]
\end{equation*}%
and $\overline{h}$ is convex in $\varepsilon _{k}$ from Lemma \ref{BL}. It
follows that also $h(\sigma _{0},\varepsilon _{0},...,\varepsilon
_{k-1},\varepsilon _{k})$ is convex in $\varepsilon _{k}$ for each value of $%
\sigma _{0},\varepsilon _{0},...,\varepsilon _{k-1}$.\bigskip
\end{proof}

We can finally state the result on the propagation of the convex order to $%
S_{n}$:

\begin{theorem}
\label{compth}Let $X_{n}$ and $S_{n}$ be as in (\ref{garch1}) and (\ref{sums}%
). Let $\varepsilon _{i}$ be symmetric. If also $\widetilde{\varepsilon }%
_{k} $ is symmetric and $\widetilde{\varepsilon }_{k}$ $\geq
_{cx}\varepsilon _{k} $, then $\widetilde{S}_{n}:=S_{n}(\sigma
_{0},\varepsilon _{0},.,\widetilde{\varepsilon }_{k},...,\varepsilon
_{n})\geq _{cx}S_{n}(\sigma _{0},\varepsilon _{0},.,\varepsilon
_{k},...,\varepsilon _{n})$.
\end{theorem}

\begin{proof}
Let $\phi $ be convex. From the independence of the $\varepsilon _{i}$ we
can write
\begin{eqnarray*}
E[\phi (\widetilde{S}_{n})] &=&E_{\varepsilon _{0},...,\varepsilon _{k-1}}E_{%
\widetilde{\varepsilon }_{k}}E_{\varepsilon _{k+1},...,\varepsilon
_{n}}[\phi (S_{n}(\sigma _{0},\varepsilon _{0},...,\varepsilon _{k-1},%
\widetilde{\varepsilon }_{k},\varepsilon _{k+1},...,\varepsilon _{n}))]= \\
&=&E_{\varepsilon _{0},...,\varepsilon _{k-1}}E_{\widetilde{\varepsilon }%
_{k}}[h(\sigma _{0},\varepsilon _{0},...,\varepsilon _{k-1},\widetilde{%
\varepsilon }_{k})]
\end{eqnarray*}%
\newline
where as in (\ref{acca})
\begin{equation*}
h(\sigma _{0},\varepsilon _{0},...,\varepsilon _{k-1},\widetilde{\varepsilon
}_{k}):=E_{\varepsilon _{k+1},...,\varepsilon _{n}}[\phi (S_{n}(\sigma
_{0},\varepsilon _{0},...,\varepsilon _{k-1},\widetilde{\varepsilon }%
_{k},\varepsilon _{k+1},...,\varepsilon _{n}))]
\end{equation*}%
is a convex function of $\widetilde{\varepsilon }_{k}$ for each value of $%
\sigma _{0},\varepsilon _{0},...,\varepsilon _{k-1}$ from Lemma \ref%
{convexity}. \newline
It follows that
\begin{equation*}
E_{\widetilde{\varepsilon }_{k}}[h(\sigma _{0},\varepsilon
_{0},...,\varepsilon _{k-1},\widetilde{\varepsilon }_{k})]\geq
E_{\varepsilon _{k}}[h(\sigma _{0},\varepsilon _{0},...,\varepsilon
_{k-1},\varepsilon _{k})]
\end{equation*}%
\newline
that gives
\begin{eqnarray*}
E[\phi (\widetilde{S}_{n})] &=&E_{\varepsilon _{0},...,\varepsilon _{k-1}}E_{%
\widetilde{\varepsilon }_{k}}[h(\sigma _{0},\varepsilon _{0},...,\varepsilon
_{k-1},\widetilde{\varepsilon }_{k})]\geq \\
&\geq &E_{\varepsilon _{0},...,\varepsilon _{k-1}}E_{\varepsilon
_{k}}[h(\sigma _{0},\varepsilon _{0},...,\varepsilon _{k-1},\varepsilon
_{k})]=E[\phi (S_{n})]
\end{eqnarray*}%
\newline
that is $\widetilde{S}_{n}\geq _{cx}S_{n}$.
\end{proof}

\section{Multivariate comparisons of logreturns}

Until now we have been considering only univariate orderings of the $X_{n}$.
Having established also a convex comparison result for the sums $S_{n}$, it
is natural to wonder whether more general multivariate comparisons do hold,
that is if we can prove that $\varepsilon _{k}\leq _{cx}\widetilde{%
\varepsilon }_{k}\Rightarrow (X_{1},...,X_{n})\preceq (\widetilde{X}_{1},...,%
\widetilde{X}_{n})$ for some multivariate convexity ordering $\preceq $ to
be precise in the sequel. Before stating a positive result, we recall two
basic definitions:

\begin{definition}
A function $\varphi :%
\mathbb{R}
^{n}\rightarrow
\mathbb{R}
$ is directionally convex\emph{\/} if for any $\mathbf{x}_{1},...,\mathbf{x}%
_{4}\in
\mathbb{R}
^{n}$, such that $\mathbf{x}_{1}\leq \mathbf{x}_{2},\mathbf{x}_{3}\leq
\mathbf{x}_{4}$ and $\mathbf{x}_{1}+\mathbf{x}_{4}=\mathbf{x}_{2}+\mathbf{x}%
_{3}$, it holds that
\begin{equation*}
\varphi (\mathbf{x}_{2})+\varphi (\mathbf{x}_{3})\leq \varphi (\mathbf{x}%
_{1})+\varphi (\mathbf{x}_{4})
\end{equation*}
\end{definition}

\begin{definition}
A function $\varphi :%
\mathbb{R}
^{n}\rightarrow
\mathbb{R}
{}$ is supermodular\emph{\/} if for any $\mathbf{x},\mathbf{y}\in
\mathbb{R}
^{n}$ it satisfies:
\begin{equation*}
\varphi (\mathbf{x})+\varphi (\mathbf{y})\leq \varphi (\mathbf{x}\wedge
\mathbf{y})+\varphi (\mathbf{x}\vee \mathbf{y}),
\end{equation*}%
where the operators $\wedge $ and $\vee $ denote respectively coordinatewise
minimum and maximum (see Section 7.A.8 of \cite{ss})
\end{definition}

In the univariate case directionally convexity is equivalent to convexity,
while in the multivariate case there are no implications between the two
concepts. Directionally convexity implies supermodularity and its equivalent
to supermodularity plus componentwise convexity. For smooth functions,
directionally convexity is equivalent to the nonnegativity of all entries in
the Hessian matrix, while supermodularity is equivalent to the the
nonnegativity of all entries out of the principal diagonal. Clearly no
implications exist between this concept and the usual convexity of $\phi $,
that corresponds to the positive semidefiniteness of the Hessian matrix.
However, $\phi $ is directionally convex and convex if and only if it is
supermodular and convex. Finally, in the smooth case, $\varphi $ is
directionally convex if and only if its gradient is increasing.

In order to establish multivariate comparison results, we need a
generalization of Lemma \ref{BL}:

\begin{lemma}
\textbf{\ \label{EBL}}Let $\phi \in C^{2}(%
\mathbb{R}
^{n})$ be convex and supermodular and $g_{i}\in C^{2}(%
\mathbb{R}
)$ be convex and nonnegative. Let $P_{n}:=\left\{ -1,1\right\} ^{n}$. It
follows that
\begin{equation*}
h(u)=\sum\limits_{\underline{p}\in P_{n}}\phi
(p_{1}g_{1}(u),...,p_{n}g_{n}(u))
\end{equation*}
is convex.
\end{lemma}

\begin{proof}
Let's denote by $y_{i}$ the arguments of the function $\phi $; we can write:%
\begin{equation*}
h^{\prime }(u)=\sum\limits_{\underline{p}\in
P_{m}}\sum\limits_{i=1}^{m}p_{i}g_{i}^{\prime }(u)\frac{\partial \phi }{%
\partial y_{i}}(p_{1}g_{1}(u),...,p_{m}g_{m}(u))
\end{equation*}

\begin{eqnarray*}
h^{\prime \prime }(u) &=&\sum\limits_{\underline{p}\in
P_{m}}[\sum\limits_{i=1}^{m}p_{i}g_{i}^{\prime \prime }(u)\frac{\partial
\phi }{\partial y_{i}}(p_{1}g_{1}(u),...,p_{m}g_{m}(u))+ \\
&&+\sum\limits_{i=1}^{m}\sum\limits_{j=1}^{m}p_{i}p_{j}g_{i}^{\prime
}(u)g_{j}^{\prime }(u)\frac{\partial ^{2}\phi }{\partial y_{i}\partial y_{j}}%
(p_{1}g_{1}(u),...,p_{m}g_{m}(u))]
\end{eqnarray*}%
The second term in the square brackets can be written as%
\begin{equation*}
\sum\limits_{i=1}^{m}\sum\limits_{j=1}^{m}p_{i}p_{j}g_{i}^{\prime
}(u)g_{j}^{\prime }(u)\frac{\partial ^{2}\phi }{\partial y_{i}\partial y_{j}}%
(p_{1}g_{1}(u),...,p_{m}g_{m}(u))=\underline{Pg}^{\prime }(u)\left[\frac{\partial
^{2}\phi }{\partial y_{i}\partial y_{j}}(..)\right](\underline{Pg}^{\prime
}(u))^{T}\geq 0
\end{equation*}%
\newline
and is positive since the Hessian of $\phi $ is positive semidefinite.%
\newline
The first term can be written as
\begin{equation*}
\sum\limits_{\underline{p}\in P_{m}}\sum\limits_{i=1}^{m}p_{i}g_{i}^{\prime
\prime }(u)\frac{\partial \phi }{\partial y_{i}}%
(p_{1}g_{1}(u),...,p_{m}g_{m}(u))=2^{m}E\left[ \nabla \phi \left(
p_{1}g_{1}(u),...,p_{n}g_{n}(u)\right) \cdot \left( \underline{Pg}^{\prime
\prime }(u)\right) \right]
\end{equation*}%
\newline
From the hypothesis on $\phi $ all the components of $\nabla \phi $ are
increasing in $p_{i}$, hence from the multivariate covariance inequality
\begin{equation*}
E\left[ \nabla \phi \left( p_{1}g_{1}(u),...,p_{n}g_{n}(u)\right) \cdot
\left( \underline{Pg}^{\prime \prime }(u)\right) \right] \geq 0
\end{equation*}
\end{proof}

As in Lemma \ref{BL}, the smoothness requirements on $\phi $ and $g_{i}$ can
be dropped, but we added them in order to simplify the proof. The
multivariate analogue of Lemma \ref{convexity} is the following:

\begin{lemma}
Let $X_{n}$ and $S_{n}$ be as in (\ref{garch1}) and (\ref{sums}). Let $\phi
:R^{n+1}\rightarrow R$ be supermodular and convex and $\varepsilon _{i}$ be
symmetric. Then the function
\begin{equation*}
h_{k}(x):=E\left[ \phi (X_{0},...,X_{n})|\varepsilon _{k}=x\right] \text{;}
\end{equation*}%
\newline
is convex.
\end{lemma}

\begin{proof}
From the symmetry of the innovations we can write
\begin{equation*}
h_{k}(x)=E_{\varepsilon _{0}...\varepsilon _{k-1}\varepsilon
_{k+1...}\varepsilon _{n}}\left[ \phi (\sigma _{0}\varepsilon
_{0},g_{1}\varepsilon _{1},...,g_{k}x,,...,g_{n}\varepsilon _{n})\right] =
\end{equation*}%
\begin{equation*}
=E_{\varepsilon _{0}...\varepsilon _{k-1}\varepsilon _{k+1...}\varepsilon
_{n}}[1_{\left\{ \underline{\varepsilon }\geq \underline{0}\right\}
}\sum\limits_{\underline{p}\in P_{n}}\phi (\sigma _{0}p_{0}\varepsilon
_{0},g_{1}p_{1}\varepsilon _{1},...,g_{k}p_{k}x,...,g_{n}p_{n}\varepsilon
_{n})]
\end{equation*}%
Since each $g_{i}$ is convex in $\varepsilon _{k}$, from Lemma \ref{EBL} it
follows that for each $\sigma _{0}>0$ and $\varepsilon _{i}\geq 0$, $i\neq k$%
, the function
\begin{equation*}
\overline{h}_{k}(x)=\sum\limits_{\underline{p}\in P_{n}}\phi (\sigma
_{0}p_{0}\varepsilon _{0},g_{1}p_{1}\varepsilon
_{1},...,g_{k}p_{k}x,...,g_{n}p_{n}\varepsilon _{n})
\end{equation*}%
is convex. Averaging with respect to $\varepsilon _{i}$, with $i\neq k$, it
follows that also $h_{k}(x)$ is convex.
\end{proof}

This enables us to state our main multivariate comparison result:

\begin{theorem}
Let $X_{n}$ and $S_{n}$ be as in (\ref{garch1}) and (\ref{sums}). Let the $%
\varepsilon _{i}$ be symmetric. If also $\widetilde{\varepsilon }_{k}$ is
symmetric and $\widetilde{\varepsilon }_{k}$ $\geq _{cx}\varepsilon _{k}$,
then,
\begin{equation*}
E\left[ \phi (X_{0},..., X_{k},...,X_{n})\right] \leq E\left[ \phi (X_{0},...,\tilde{X}%
_{k},...,\tilde{X}_{n})\right]
\end{equation*}
for every function $\phi :R^{n+1}\rightarrow R$ supermodular and convex.
\end{theorem}

\begin{proof}
From the previous lemma we have that
\begin{equation*}
E\left[ \phi (X_{0},..., X_{k},...,X_{n})\right] =E_{\varepsilon _{k}}[h_{k}(\varepsilon
_{k})]\leq E_{\widetilde{\varepsilon }_{k}}[h_{k}(\widetilde{\varepsilon }%
_{k})]=E\left[ \phi (X_{0},...\tilde{X}_{k},...,\tilde{X}_{n})\right] \text{.}
\end{equation*}
\end{proof}

We remark that we are not able to prove supermodularity or componentwise
convex ordering of $(X_{0},...,X_{n})$; at the moment both hypothesys on $%
\phi $ seem to be necessary for Lemma \ref{EBL}.

\section{The Garch (1,1) case}

We focus now on the Garch (1,1) model specified by

\begin{equation}
\left\{
\begin{array}{c}
X_{n}=\sigma _{n}\varepsilon _{n} \\
\sigma _{n+1}^{2}=\alpha _{0}+\alpha _{1}X_{n}^{2}+\beta _{1}\sigma _{n}^{2}%
\end{array}%
\right. \text{, }  \label{garch}
\end{equation}%
\newline
with $\alpha _{0},\alpha _{1},\beta _{1}>0$ and $\alpha _{1}+\beta _{1}<1$.
For this model the recursive dynamic of the volatility or of the variance (%
\ref{gienne2}) can easily be explicitated as follows (see \cite{nelson}):

\begin{equation}
\sigma _{n+1}^{2}=\sigma _{0}^{2}\prod\limits_{i=1}^{n+1}(\beta _{1}+\alpha
_{1}\varepsilon _{n-i+1}^{2})+\alpha _{0}\left[ 1+\sum\limits_{k=1}^{n}\prod%
\limits_{i=1}^{k}(\beta _{1}+\alpha _{1}\varepsilon _{n-i+1}^{2})\right]
\label{varnelson}
\end{equation}%
\newline
From this expression it is immediate that $\sigma _{n+1}^{2}$ and $\sigma
_{n+1}$ are nondecreasing functions of the parameters $\alpha _{0}$, $\alpha
_{1}$and $\beta _{1}$.We already remarked that this model is a special case
of both M1 and M2, so all the comparison result for varying innovations of
the preceding sections do hold. In this section we are interested in
establishing comparison results for different parameters $\alpha _{0}$, $%
\alpha _{1}$, $\beta _{1}$. As mentioned in the Introduction, it is natural
that an increase in $\alpha _{0}$, $\alpha _{1}$, $\beta _{1}$ should
correspond to an increase in the variability of $X_{n}$ and $S_{n}$; in this
section we prove it rigorously. Without any additional effort, we can
consider stochastic parameters $\alpha _{0}$, $\alpha _{1}$, $\beta _{1}$:

\begin{proposition}
Let $X_{n}$ be as in (\ref{garch}). If we consider random parameters $\alpha
_{0}\leq _{st}\widetilde{\alpha }_{0}$, $\alpha _{1}\leq _{st}\widetilde{%
\alpha }_{1}$ and $\beta _{1}\leq _{st}\widetilde{\beta }_{1}$, then $%
\left\vert X_{n}\right\vert \leq _{st}\left\vert \widetilde{X}%
_{n}\right\vert $, $X_{n}^{2}\leq _{st}\widetilde{X}_{n}^{2}$ and $%
X_{n}\preceq _{cx}\widetilde{X}_{n}$.
\end{proposition}

\begin{proof}
Since $\sigma _{n}$ and $\sigma _{n}^{2}$are increasing functions of the
parameters, if $\alpha _{0}\leq _{st}\widetilde{\alpha }_{0}$, $\alpha
_{1}\leq _{st}\widetilde{\alpha }_{1}$ and $\beta _{1}\leq _{st}\widetilde{%
\beta }_{1}$ it follows that $\sigma _{n}\leq _{st}\widetilde{\sigma }_{n}$
and $\sigma _{n}^{2}\leq _{st}\widetilde{\sigma }_{n}^{2}$. As in the proof
of Theorem \ref{compX} it follows that $\left\vert X_{n}\right\vert \leq
_{st}\left\vert \widetilde{X}_{n}\right\vert $and $X_{n}^{2}\leq _{st}%
\widetilde{X}_{n}^{2}$. From Lemma \ref{stconv} $\sigma _{n}\leq _{st}%
\widetilde{\sigma }_{n}$ implies that $X_{n}\preceq _{cx}\widetilde{X}_{n}$.
\end{proof}

The last point is to prove the convex comparison of the sums $S_{n}$; again,
this is nontrivial since the $X_{n}$ are not independent; we provide a proof
in the case of symmetric innovations.

\begin{theorem}
Let $X_{n}$ be as in (\ref{garch}) and\ $S_{n}$ as in (\ref{sums}). Let $%
\varepsilon _{i}$ be symmetric. If we consider random parameters $\alpha
_{0}\leq _{st}\widetilde{\alpha }_{0}$, $\alpha _{1}\leq _{st}\widetilde{%
\alpha }_{1}$ and $\beta _{1}\leq _{st}\widetilde{\beta }_{1}$, then $%
S_{n}\preceq _{cx}\tilde{S}_{n}$.
\end{theorem}

\begin{proof}
As before, we write%
\begin{equation*}
S_{n}=\sigma _{0}\varepsilon _{0}+g_{1}(\varepsilon _{0},...,\alpha
_{0},\alpha _{1},\beta _{1})\varepsilon _{1}+...+g_{n}(\varepsilon
_{0},...,\alpha _{0},\alpha _{1},\beta _{1})\varepsilon _{n}
\end{equation*}%
where the functions $g_{i}$ are nondecreasing in the parameters $\alpha
_{0},\alpha _{1},\beta _{1}$. Let $\phi $ be any convex function. We first
want to prove that $E\left[ \phi \left( S_{n}\right) \right] $ is
nondecreasing in the parameters $\alpha _{0},\alpha _{1},\beta _{1}$. From
the symmetry of the innovations $\varepsilon _{i}$ we can write:%
\begin{equation*}
E\left[ \phi \left( S_{n}\right) \right] =E_{\varepsilon
_{0},...,\varepsilon _{n}}[\phi (\sigma _{0}\varepsilon
_{0}+...+g_{n}\varepsilon _{n})]
\end{equation*}%
\begin{equation*}
=E_{\varepsilon _{0},...,\varepsilon _{n}}[\sum\limits_{\underline{p}\in
P_{n+1}}\phi (\sigma _{0}p_{0}\varepsilon _{0}+...+p_{n}g_{n}\varepsilon
_{n})1_{\left\{ \varepsilon _{0}\geq 0,...,\varepsilon _{n}\geq 0\right\} }]%
\text{.}
\end{equation*}%
Denoting by
\begin{equation*}
\overline{h}(\sigma _{0},...,\alpha _{0},\alpha _{1},\beta
_{1})=\sum\limits_{\underline{p}\in P_{n+1}}\phi (\sigma
_{0}p_{0}\varepsilon _{0}+...+p_{n}g_{n}\varepsilon _{n})\text{,}
\end{equation*}%
\newline
we see that $\overline{h}$ is nondecreasing in $\alpha _{0},\alpha
_{1},\beta _{1}$; indeed we can compute:
\begin{equation*}
\frac{\partial \overline{h}}{\partial \alpha _{0}}=\sum\limits_{\underline{p}%
\in P_{n+1}}\phi ^{\prime }\left(\sigma
_{0}p_{0}\varepsilon_{0}+...+p_{n}g_{n}\varepsilon _{n}\right) \cdot (p_{1}\varepsilon _{1}g_{1}^{\prime }+...+p_{n}\varepsilon _{n}g_{n}^{\prime })\geq 0
\end{equation*}%
\newline
from the multivariate covariance inequality, as in the proof of Lemma \ref%
{BL}. The same reasoning shows that $\frac{\partial \overline{h}}{\partial
\alpha _{1}}\geq 0$ and $\frac{\partial \overline{h}}{\partial \beta _{1}}%
\geq 0$.\newline
It follows that $E\left[ \phi \left( S_{n}\right) \right] $ is is
nondecreasing in $\alpha _{0},\alpha _{1},\beta _{1}$; but then if $\alpha
_{0}\leq _{st}\widetilde{\alpha }_{0}$, $\alpha _{1}\leq _{st}\widetilde{%
\alpha }_{1}$ and $\beta _{1}\leq _{st}\widetilde{\beta }_{1}$,
\begin{equation*}
E\left[ \phi \left( S_{n}(\alpha _{0},\alpha _{1},\beta _{1}\right) \right]
\leq E[\phi (S_{n}(\tilde{\alpha}_{0},\tilde{\alpha}_{1},\tilde{\beta}_{1})]
\end{equation*}%
that is $S_{n}\preceq _{cx}\tilde{S}_{n}$.
\end{proof}

\end{document}